\documentclass[10pt,conference]{IEEEtran}
\usepackage{amsmath,amsfonts}
\usepackage{mathrsfs}
\usepackage{bbm}
\usepackage{algorithmic}
\usepackage{array}
\usepackage{subfig}
\usepackage{textcomp}
\usepackage{stfloats}
\usepackage{url}
\usepackage{verbatim}
\usepackage{graphicx}   
\usepackage{booktabs}   
\usepackage{tikz}
\usetikzlibrary{arrows.meta,positioning,fit,calc}

\usepackage{cite}
\usepackage{qcircuit}
\usepackage{amssymb}
\usepackage{amsthm}
\usepackage{braket}
\usepackage[linesnumbered, ruled, vlined, boxed]{algorithm2e}
\usepackage{xcolor}
\usepackage{textcomp}
\usepackage{multirow}
\usepackage{makecell}
\usepackage{diagbox}
\usepackage{comment}
\usepackage{svg}
\usepackage[table]{xcolor}

\newtheorem{theorem}{Theorem}
\newtheorem{definition}{Definition}

\newtheorem{proposition}{Proposition}

\def\>{\ensuremath{\rangle}}
\def\<{\ensuremath{\langle}}

\def\BibTeX{{\rm B\kern-.05em{\sc i\kern-.025em b}\kern-.08em
    T\kern-.1667em\lower.7ex\hbox{E}\kern-.125emX}}
\usepackage{balance}

\begin{document}

\title{SOFT: a high-performance simulator for universal fault-tolerant quantum circuits}
\author{%
  Riling Li$^{1,*, \dagger}$,
  Keli Zheng$^{2,*}$,
  Yiming Zhang$^{3,4}$,
  Huazhe Lou$^{5}$,
  Shenggang Ying$^{1}$,
  Ke Liu$^{3,4, \dagger}$,
  Xiaoming Sun$^{6, \dagger}$\\
  \IEEEauthorblockA{$^{1}$Institute of Software, Chinese Academy of Sciences}
  \IEEEauthorblockA{$^{2}$The University of Hong Kong}
  \IEEEauthorblockA{$^{3}$Hefei National Research Center for Physical Sciences at the Microscale and School of Physical Sciences,\\ University of Science and Technology of China}
  \IEEEauthorblockA{$^{4}$Shanghai Research Center for Quantum Science\\ and CAS Center for Excellence in Quantum Information and Quantum Physics,\\ University of Science and Technology of China}
  \IEEEauthorblockA{$^{5}$Arclight Quantum Corporation}
  \IEEEauthorblockA{$^{6}$Institute of Computing Technology,
Chinese Academy of Sciences}\\
  \IEEEauthorblockA{$^{*}$Equal contribution.}
  \IEEEauthorblockA{$^{\dagger}$Corresponding author: \texttt{haoliri0@gmail.com, ke.liu@ustc.edu.cn, sunxiaoming@ict.ac.cn}.}
}

\maketitle

\begin{abstract}
Circuit simulation tools are critical for developing and assessing quantum-error-correcting and fault-tolerant strategies. In this work, we present SOFT, a high-performance SimulatOr for universal Fault-Tolerant quantum circuits. Integrating the generalized stabilizer formalism and highly optimized GPU parallelization, SOFT enables the simulation of noisy quantum circuits containing non-Clifford gates at a scale not accessible with existing tools. To provide a concrete demonstration, we simulate the state-of-the-art magic state cultivation (MSC) protocol at code distance $d=5$, involving 42 qubits, 72 $T$ / $T^\dagger$ gates, and mid-circuit measurements. Using only modest GPU resources, SOFT performs over 200 billion shots and achieves the first ground-truth simulation of the cultivation protocol at a non-trivial scale. 
This endeavor not only certifies the MSC's effectiveness for generating high-fidelity logical $T$-states, but also reveals a large discrepancy between the actual logical error rate and the previously reported values. 
Our work demonstrates the importance of reliable simulation tools for fault-tolerant architecture design, advancing the field from simulating quantum memory to simulating a universal quantum computer.
\end{abstract}


\section{Introduction}
The development of a fully-fledged quantum computer fundamentally relies on quantum error correction (QEC) and fault-tolerance (FT) to protect quantum information from decoherence~\cite{Dennis02, Shor95, Steane96,Gottesman98}.
Driven by rapid advances in hardware, various quantum platforms—including superconducting systems, neutral atoms, and trapped ions—have recently realized functional logical qubits, successfully demonstrating the suppression of logical error rates~\cite{Google25,He25,Moses23,Bluvstein24}.
Consequently, the field is accelerating its transition from realizing quantum memory to achieving universal quantum computation. A central challenge in this pursuit is the fault-tolerant implementation of logical non-Clifford gates, such as the $T$ gate and Toffoli gate. While these operations underpin the theoretical advantage of quantum algorithms, they also dominate the qubit and gate overhead in practical applications. Therefore, substantial contemporary research is devoted to developing low-overhead fault-tolerant protocols, particularly to reduce the high costs of logical non-Clifford gates.

Nonetheless, despite the emergence of various potentially efficient proposals for implementing logical non-Clifford gates with noisy qubits~\cite{Litinski19,Chamberland20,gidney2024magic}, reliably assessing their validity and performance remains a crucial bottleneck. 
The difficulty originates from the fundamental limitations of simulating non-Clifford gates on classical computers. As implied by the Gottesman-Knill theorem~\cite{gottesman1997stabilizer}, the presence of non-Clifford operations, such as the $T$ gate, causes the simulation complexity to scale exponentially with the gate count. 
Consequently, researchers typically resort to empirical conjectures or uncontrolled approximations to estimate the performance of sophisticated protocols. This inability to establish ground-truth assessment creates a critical gap in resource estimation and architectural design, impeding the development of resource-efficient fault-tolerant quantum computation.

To address this gap, we present SOFT—a high-performance \textbf{S}imulat\textbf{O}r for universal \textbf{F}ault-\textbf{T}olerant quantum circuits.
We focus here on noisy Clifford $+T$ circuits, which represents the most common basis for universal quantum computation and dominates the implementation protocols of non-Clifford gates. SOFT integrates the generalized stabilizer formalism with GPU acceleration, with a fully device-resident shot-parallel design that preallocates per-shot buffers and updates the generalized stabilizer tableau via lightweight staged kernels.
Moreover, circuit-level noise and mid-circuit measurements are faithfully accounted to reflect the conditions of noisy quantum hardware. 
The presence of noise renders the simulation more challenging than the noiseless case and makes efficient sampling even more crucial.
 Our simulation strategy boosts the simulation throughput by {\it several orders} of magnitude compared to existing tools, thus making it possible for the first time to certify non-Clifford logical gates targeting extremely high fidelity.

As a concrete demonstration, we apply SOFT to the magic state cultivation (MSC) protocol~\cite{gidney2024magic}, a recently proposed method for efficiently preparing high-fidelity logical $T$ states. 
The distance-3 ($d=3$) cultivation, comprising 15 qubits and 22 $T$/$T^\dagger$ gates, is amenable to brute-force simulation and has recently been experimentally realized on Google's Willow processor~\cite{rosenfeld2025magic}. 
However, scaling the circuit to distance-5, which contains 42 qubits and 72 $T$/$T^\dagger$ gates, is already computationally intractable with existing tools; hence its performance was estimated based on an empirical Clifford approximation~\cite{gidney2024magic}.
With SOFT, we now accurately simulate the $d=5$ cultivation.

We simulate the $d=5$ noisy cultivation circuit and collect over 200 billion samples, using 16 NVIDIA H800 GPUs within 20 days.
Leveraging this ground-truth simulation capability and extensive statistical data, we are able to reliably analyze fault propagation and state fidelity.
 We find that $d=5$ cultivation indeed generates high-fidelity logical $T$-states, thereby confirming the effectiveness and scalability of the  protocol. 
 However, their actual fidelity differs significantly from previously reported values~\cite{gidney2024magic}.
At physical error rates of $p=10^{-3}$ and $5\times 10^{-4}$, which are relevant to near-term quantum hardware, the discrepancy factor is nearly an order of magnitude, as summarized in Table I. 
We also observe that this discrepancy widens at larger circuit sizes and lower noise levels—precisely the regimes where an extremely large number of samples is required to certify high-fidelity non-Clifford states. 
Therefore, our work not only provides a definitive validation of the MSC protocol but also prominently highlights the crucial distinction between simulating quantum memories and simulating universal quantum computers.

\begin{table*}
    \centering
    \caption{Our reliable simulations vs. empirical conjecture in Ref.~\cite{gidney2024magic}.}
    \label{tab1}
    {\begin{tabular}{c|c|c|c|c|c|c}
      \hline
     Noise Strength & Method & Logical Error Rate & Discrepancy Factor & Total Shots & Preserved Shots & Detected Logical Errors \\
      \Xhline{1.5pt}
    \multirow{2}{*}{0.002}  & Reliable simulation  & $3.41\times 10^{-8}$ & 1.87 & 31.0 billion & 0.64 billion  & 22  \\
      \cline{2-7}
    &  Empirical conjecture & $ 1.82\times 10^{-8} $ & / & N/A  & N/A  & N/A \\
    \hline
    \multirow{2}{*}{0.001}  & Reliable simulation  & $4.59\times 10^{-9}$ & 7.65 & 74.0 billion & 10.6 billion  & 49  \\
      \cline{2-7}
    &  Empirical conjecture & $6\times 10^{-10}$ & / & N/A  & N/A  & N/A \\
      \hline
    \multirow{2}{*}{0.0005}  & Reliable simulation  & $1.57\times 10^{-10}$ & 7.85 & 134.4 billion  & 50.9 billion  & 8  \\
      \cline{2-7}
    &  Empirical conjecture & $2\times 10^{-11}$ & / & N/A  & N/A  & N/A \\
      \hline
    \end{tabular}}
  \end{table*}

\subsection{Related Works}
This study focuses on quantum circuit simulation, an area that has seen significant development over the past decade.
The most commonly used method is the Schrödinger-style state-vector simulation, where the entire quantum state is represented as a vector that evolves under gate operations.
However, the required memory grows exponentially with the number of qubits, limiting simulations to approximately 45 qubits on modern supercomputing hardware~\cite{haner2017, Sunway2019massively, Sunway2019quantum, de2025universal}.

Another major category of methods is based on tensor-network \cite{huang2020classical, liu2021GB, pan2022solving, fu2024surpassing}.
These approaches are not directly limited by the number of qubits or circuit depth but are instead constrained by a structural parameter known as treewidth.
Tensor-network variants such as MPS~\cite{ayral2023density} (Matrix Product States) and PEPS~\cite{patra2024efficient} (Projected Entangled Pair States) perform well on circuits with specific structures but become inefficient for highly entangled systems, making them less suitable for quantum error-correcting circuits.

A third family of methods relies on decision diagrams (DDs), which compress the representation of quantum states \cite{hong2020tensor, vinkhuijzen2023limdd, sistla2023symbolic}.
While DD-based methods can achieve compact representations for certain structured circuits, they provide limited compression for random or highly entangled circuits.
Furthermore, their complex data structures hinder parallelization, making most DD-based simulators relatively inefficient.

For quantum error-correction circuits, the Gottesman–Knill theorem states that stabilizer circuits composed of CNOT, Hadamard, and phase (S) gates, together with Pauli measurements, can be simulated efficiently in polynomial time.
Hence, simulators based on the stabilizer formalism are particularly suitable for Clifford circuits.
Notable high-performance implementations include Stim~\cite{gidney2021stim} and SymPhase~\cite{fang2024symphase}, which have extensive algorithmic and implementation-level optimizations.
However, these simulators cannot handle fault-tolerant circuits containing dozens of non-Clifford gates. 

To overcome this limitation, stabilizer decomposition attracted attention \cite{bravyi2016improved, bravyi2019simulation}, but these early works focused on computing a single output amplitude and paid relatively little attention to mid-circuit measurements. Recent studies have also proposed hybrid stabilizer–tensor-network approaches, combining the strengths of both methods~\cite{masot2024stn, nakhl2025stn}.
Nonetheless, existing implementations remain inefficient and lack a detailed structural analysis of fault-tolerant circuits. 

In \cite{rosenfeld2025magic}, the \texttt{qevol} simulator is introduced to simulate $d=5$ grafting experiments following a $d=3$ injection+cultivation stage. While their approach similarly employs a stabilizer decomposition method, it does not extend to the $d=5$ cultivation stage. The latter is computationally more demanding than the grafting (escape) stage, which consists exclusively of Clifford operations and is therefore significantly less expensive to simulate.


\section{Preliminaries}
\subsection{Clifford gates, $T$ gates, and Pauli measurements}
Let $\mathcal{P}_n=\{\pm 1,\pm i\}\{I,X,Y,Z\}^{\otimes n}$ be the $n$-qubit Pauli group.  
The $n$-qubit Clifford group is defined as
\[
\mathcal{C}_n \;=\; \{\,U \mid U \mathcal{P}_n U^\dagger = \mathcal{P}_n\,\}.
\]
In this work we use the standard generating set $\{H,S,\mathrm{CNOT}\}$ for Clifford circuits.

A universal gate set is obtained by adding a non-Clifford gate, most commonly the $\pi/8$ gate $T$ (and its inverse $T^\dagger$). 
Up to a global phase, $T$ and $T^{\dagger}$ are $Z$-rotation and admit the Clifford-linear form
\begin{equation}\label{eq:T_gate}
T \;=\; e^{i\pi/8}\!\left(\cos\frac{\pi}{8}\,I - i\sin\frac{\pi}{8}\,Z\right),
\end{equation}
\begin{equation}\label{eq:T_dagger}
T^\dagger \;=\; e^{-i\pi/8}\!\left(\cos\frac{\pi}{8}\,I + i\sin\frac{\pi}{8}\,Z\right),
\end{equation}

which will be useful for our update rules.

We consider (possibly mid-circuit) Pauli measurements of observables $P\in \mathcal{P}_n$ with outcomes $\pm1$, implemented by projectors
\[
\Pi_{\pm}=\frac{I\pm P}{2}.
\]
Such measurements are central to stabilizer-based error correction: the outcomes reveal syndromes while the post-measurement state remains efficiently representable in the stabilizer tableau formalism (see below).

\subsection{Stabilizer and destabilizer tableau}
The stabilizer formalism~\cite{gottesman1997stabilizer} provides an efficient representation of \emph{stabilizer states}. 
An $n$-qubit stabilizer state $|\psi_{\mathcal{S}}\rangle$ is the unique simultaneous $+1$ eigenstate of an Abelian subgroup $\mathcal{S}\subset \mathcal{P}_n$ that does not contain $-I$. 
Equivalently, $\mathcal{S}$ can be specified by $n$ independent commuting generators
\[
\mathcal{S}=\langle s_1,\dots,s_n\rangle,\qquad
s_i|\psi_{\mathcal{S}}\rangle=|\psi_{\mathcal{S}}\rangle \ \mathrm{for\ 1\leq i \leq n}.
\]
Under a Clifford gate $U\in\mathcal{C}_n$, stabilizers update by conjugation:
\[
s_i \;\mapsto\; s_i^{\prime}=U s_i U^\dagger, \ \mathrm{for\ 1\leq i \leq n}
\]
so Clifford circuits with Pauli measurements can be simulated with polynomial-time tableau updates.

In practice, it is convenient to maintain, in addition to stabilizers, a complementary set of \emph{destabilizers}~\cite{aaronson2004improved}
\[
\mathcal{D}=\{d_1,\dots,d_n\}\subset \mathcal{P}_n,
\]
chosen such that $d_i$ anticommutes with $s_i$ and commutes with $s_j$ for $j\neq i$:
\[
d_i s_j = (-1)^{\delta_{ij}} s_j d_i.
\]
The pair $(\mathcal{S},\mathcal{D})$ forms the stabilizer--destabilizer tableau, which supports efficient binary-symplectic update rules for both Clifford conjugation and Pauli measurements. 
This tableau structure underlies many modern stabilizer simulators, including Stim~\cite{gidney2021stim} and Symphase~\cite{fang2024symphase}.

\subsection{Generalized stabilizer tableau}
\label{subsec:gs_tableau}

The stabilizer formalism efficiently tracks Clifford dynamics, but it is not closed under non-Clifford gates such as $T$.
The \emph{generalized stabilizer formalism}~\cite{yoder2012generalization} extends the stabilizer--destabilizer tableau by attaching complex coefficients, yielding a compact representation for arbitrary states.

\paragraph{Basis induced by a tableau}
Let $(\mathcal{S},\mathcal{D})$ be a stabilizer--destabilizer tableau on $n$ qubits, with
\[
\mathcal{S}=\langle s_{1},\dots,s_{n}\rangle,\qquad
\mathcal{D}=\{d_{1},\dots,d_{n}\}.
\]
For each $\alpha=(\alpha_1,\dots,\alpha_n)\in\{0,1\}^n$, define the Pauli operator
\[
\mathbf{d}_\alpha := \prod_{k=1}^n \bigl(d_k\bigr)^{\alpha_k},
\qquad
|b_\alpha\rangle := \mathbf{d}_\alpha\,|\psi_{\mathcal{S}}\rangle,
\]
where $|\psi_{\mathcal{S}}\rangle$ is the stabilizer state stabilized by $\mathcal{S}$.
The set $B(\mathcal{S},\mathcal{D})=\{|b_\alpha\rangle:\alpha\in\{0,1\}^n\}$ forms an orthonormal basis of $\mathcal{H}_2^{\otimes n}$.

\begin{definition}[Generalized stabilizer representation~\cite{yoder2012generalization}]
\label{def:gsf}
A pure $n$-qubit state $|\phi\rangle$ is represented by a pair $(\mathbf{v},B(\mathcal{S},\mathcal{D}))$, where
\begin{equation}
\label{eq:gsf_pure}
|\phi\rangle \;=\; \sum_{\alpha\in\{0,1\}^n} v_\alpha\,|b_\alpha\rangle,
\qquad
\sum_\alpha |v_\alpha|^2 = 1.
\end{equation}
We denote by $|\mathbf{v}|$ the number of nonzero coefficients $v_\alpha$.
\end{definition}

\paragraph{Pure-state simulation model}
Throughout this work we simulate stochastic Pauli noise and mid-circuit measurements by Monte Carlo sampling: 
noise is applied by randomly sampling Pauli errors, and a Pauli measurement is simulated by sampling an outcome according to its branch probability and renormalizing the post-measurement state.
Hence each shot maintains a \emph{pure} state, and we only use the pure-state representation in Equation~\eqref{eq:gsf_pure}.

\paragraph{Storage cost}
The tableau $(\mathcal{S},\mathcal{D})$ requires $O(n^2)$ bits in binary symplectic form.
The coefficient vector $\mathbf{v}$ lives in dimension $2^n$, but in practice we store it sparsely with $|\mathbf{v}|$ nonzero entries.

\paragraph{A useful identity}
For any Pauli $Q\in\mathcal{P}_n$, its action on the basis is a permutation up to phase:
\begin{equation}
\label{eq:pauli_on_basis}
Q\,|b_\alpha\rangle \;=\; \xi_{\alpha}(Q)\,|b_{\alpha\oplus \mathbf{\beta}(Q)}\rangle,
\qquad \xi_{\alpha}(Q)\in\{\pm 1,\pm i\},
\end{equation}
where the \emph{index shift} $\mathbf{\beta}(Q)\in\{0,1\}^n$ is defined by
\[
\beta_i(Q)=1 \iff Q \text{ anticommutes with } s_i.
\]
Given the tableau, $\beta(Q)$ can be computed in $O(n)$ time.

\begin{theorem}[Update cost and sparsity bound~\cite{yoder2012generalization}]
\label{thm:complexity}
Let $|\phi\rangle$ be a pure $n$-qubit state represented as $(\mathbf{v},B(\mathcal{S},\mathcal{D}))$, and let $|\mathbf{v}|$ be the number of nonzero coefficients.
For a single operation, the update cost and the change in $|\mathbf{v}|$ are:
\begin{enumerate}
    \item \textbf{Clifford gate $U\in\mathcal{C}_n$:} update cost $O(n)$ and $|\mathbf{v}'|=|\mathbf{v}|$.
    \item \textbf{Pauli measurement of $P\in\mathcal{P}_n$:} update cost $O(|\mathbf{v}|\,n+n^2)$ and $|\mathbf{v}'|\le |\mathbf{v}|$.
    \item \textbf{$T_q$ or $T_q^\dagger$ on qubit $q$:} update cost $O(|\mathbf{v}|\,n+n^2)$ and $|\mathbf{v}'|\le 2|\mathbf{v}|$.
\end{enumerate}
\end{theorem}

\begin{proof}
We sketch the key points; all costs are in binary operations on the tableau plus arithmetic on the nonzero entries of $\mathbf{v}$.

\textbf{(1) Clifford gates.}
For $U\in\mathcal{C}_n$, the tableau updates by conjugation
$(\mathcal{S},\mathcal{D})\mapsto(U\mathcal{S}U^\dagger,U\mathcal{D}U^\dagger)$,
which costs $O(n)$ for the standard generators (each gate touches $O(1)$ qubits and updates all $2n$ rows).
If we update the basis accordingly, then $U|b_\alpha\rangle=|b'_\alpha\rangle$ for all $\alpha$, so the coefficient vector is unchanged and $|\mathbf{v}'|=|\mathbf{v}|$.

\textbf{(2) Pauli measurements.}
Measuring $P$ uses projectors $\Pi_{\pm}=(I\pm P)/2$.
Using~\eqref{eq:pauli_on_basis}, for each nonzero $v_\alpha$ we can determine in $O(n)$ time whether $|b_\alpha\rangle$ lies in the sampled measurement branch (equivalently, whether it has the correct eigenvalue after the standard tableau pivot). 
Thus we filter (and renormalize) the existing nonzero entries without creating new ones, giving $|\mathbf{v}'|\le|\mathbf{v}|$ and $O(|\mathbf{v}|\,n)$ work.
The tableau update for a non-deterministic Pauli measurement requires a pivot/elimination step costing $O(n^2)$ in the Aaronson--Gottesman tableau algorithm~\cite{aaronson2004improved}.
Overall: $O(|\mathbf{v}|\,n+n^2)$.

\textbf{(3) $T$ and $T^\dagger$.}
Using the decomposition in Equation~(\ref{eq:T_gate}) and~(\ref{eq:T_dagger}), each basis vector maps to at most two basis vectors by~\eqref{eq:pauli_on_basis} (with $Q=Z_q$):
\[
T_q|b_\alpha\rangle = a\,|b_\alpha\rangle + b\,\xi_{\alpha}(Z_q)\,|b_{\alpha\oplus\mathbf{\beta}(Z_q)}\rangle,
\]
and similarly for $T_q^\dagger$.
Hence each nonzero coefficient contributes to at most two outputs, so $|\mathbf{v}'|\le 2|\mathbf{v}|$.
Computing the target index $\alpha\oplus\beta(Z_q)$ and accumulating amplitudes costs $O(n)$ per nonzero entry, giving $O(|\mathbf{v}|\,n)$.
In addition, we may apply an $O(n^2)$ re-gauging/canonicalization step on the tableau (row/column operations) when required by the implementation.
Overall: $O(|\mathbf{v}|\,n+n^2)$.
\end{proof}

\section{Magic State Cultivation}

In this section, we review the magic state cultivation
protocol~\cite{gidney2024magic} with modifications. Their scheme prepares the $|T\rangle_L$ state on a triangular color code and uses the transversal $H_{XY}$ gate to perform logical checks that progressively increase the fault distance. The preparation procedure consists of three stages: injection, cultivation, and escape.
During the injection and cultivation stages, there are Clifford gates, non-Clifford gates, Pauli measurements and detector parity checkers. See Figure \ref{fig1} for a part of the $d=3$ MSC circuit.
In \cite{gidney2024magic}, they performed state-vector simulations for the $d=3$ injection and cultivation stages, but substituted all $T (T^\dagger)$ gates with $S (S^\dagger)$ to enable a Clifford-only proxy simulation for the $d=5$ case. 
\begin{figure*}
      \centering
      \includegraphics[width=0.95\linewidth]{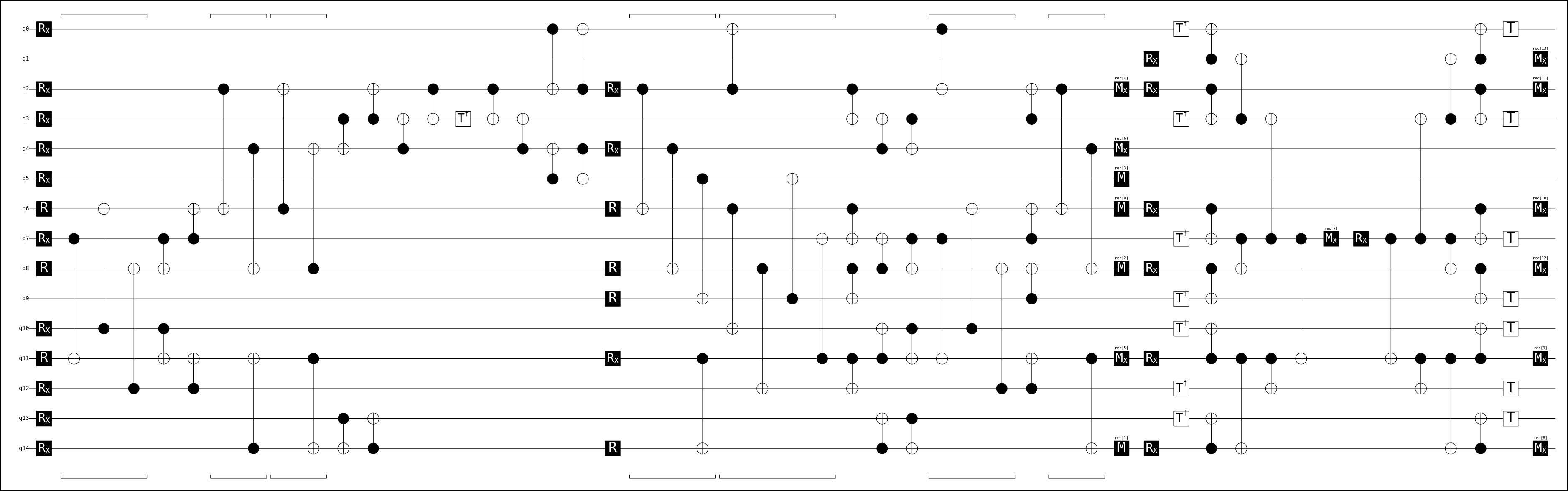}
      \caption{$d=3$ MSC circuit excluding the noise channels and the final Pauli product measurements.}
      \label{fig1}
\end{figure*}
We make several modifications to the $d=5$ injection and cultivation circuit to correct the errors noted in Section~2.6 of the original paper. Throughout the circuit, we enforce that the state remains in the $+1$ eigenstate of all 4-body and 6-body stabilizers of the color code. To ensure this, we apply the appropriate Pauli corrections after growing the code from $d=3$ to $d=5$. Assuming that the code grows along the red boundary, we initialize Bell pairs on all red edges outside the original $d=3$ patch. During the subsequent round of the superdense error-correction cycle, all ancilla qubits within red tiles will produce random measurement outcomes. Due to the structure of the superdense cycle, this implies that the eigenvalues of the $X$ stabilizers on the newly grown red tiles, as well as the $Z$ stabilizers on the grown blue or green tiles that vertically border these red tiles, will be random. In our simulations, we apply Pauli $X$ and $Z$ corrections conditioned on these measurement results to restore the stabilizers to their $+1$ eigenstate. Some of these Pauli corrections flip the logical check values during cultivation, and we account for this effect in the logical check detector definition. The transversal logical $H_{XY}$ gate is now implemented as physical $H_{XY}$ and $H_{NXY}$ gates on the data qubits. Consequently, when conjugating the $H_{XY}$ check into an $X$ check during cultivation, certain $T \,(T^\dagger)$ rotations must be replaced with $T^\dagger \,(T)$ gates.

For the noisy simulation, we use the \emph{uniform depolarizing} noise model, which inserts depolarizing noise after all physical gates, during qubit idling, and applies flip noise before measurements and after resets. All noise processes are parameterized by the uniform noise strength $p$.

The $d=5$ injection and cultivation circuit has 42 qubits and 72 non-Clifford ($T$ or $T^\dagger$) gates, as detailed in Table \ref{tab2}. Note that we do not simulate the escape stage of the protocol; however, it should be straightforward to include in future work, as all operations during the escape stage are Clifford operations and classical procedures such as complementary gap decoding.

\definecolor{zebraodd}{RGB}{255, 255, 255}
\definecolor{zebraeven}{RGB}{245, 245, 245}

\begin{table}[ht]
    \centering
    \caption{Circuit Statistics for Magic State Cultivation}
    \label{tab2}
    \small
    \renewcommand{\arraystretch}{1.3}
    
    \rowcolors{2}{zebraeven}{zebraodd}
    
    \begin{tabular}{lrr}
        \toprule
        \textbf{Metric} & \textbf{$d=3$} & \textbf{$d=5$} \\
        \midrule
        Total Qubits                         & 15       & 42       \\
        Total Gates                          & 137      & 741      \\
        Circuit Depth                        & 39       & 94       \\
        Two-Qubit Gates                      & 81       & 477      \\
        Measurements                         & 14       & 93       \\
        $T$/$T^\dagger$ Count                & 22       & 72       \\
        $T$/$T^\dagger$ Support Size         & 7        & 19       \\
        $T$/$T^\dagger$ Depth                & 4        & 6        \\
        \bottomrule
    \end{tabular}
\end{table}

\section{Method and Implementation}
\subsection{Analysis of MSC Circuit Structure}
The analysis in Section~\ref{subsec:gs_tableau} indicates that, to simulate a generic Clifford+$T$ circuit using generalized stabilizer representation $(\mathbf{v},B(\mathcal{S},\mathcal{D}))$, the number of coefficients $|v|$ in the representation can grow exponentially with the number of $T$ gates, and in the worst case it can reach $\mathcal{O}(2^n)$.\footnote{More precisely, one typically expects an upper bound of the form $\mathcal{O}(2^{\min(t,n)})$, where $t$ is the $T$-count.} If such worst-case behavior occurred for the $d=5$ MSC circuit, then \textsc{SOFT} would not be able to simulate it efficiently nor generate a large number of samples.

However, we obtain the following observations:

\begin{enumerate}
    \item In the $d=5$ MSC circuit, all $T/T^\dagger$ gates act on only $19$ qubits. More specifically, these are exactly the $19$ data qubits of the $[[19,1,5]]$ color code.
    \item The circuit contains only six layers of $T/T^\dagger$ gates, and between any two consecutive layers there are many rounds of measurement and reset operations that reinitialize a subset of qubits (typically ancillas).
    \item If $Z_q$ is contained in the stabilizer group of the current stabilizer component, i.e.,
    $Z_q\ket{\Psi_\mathcal{S}}=\pm\ket{\Psi_\mathcal{S}}$,
    then applying $T$ or $T^\dagger$ on qubit $q$ does not increase $|\mathbf{v}|$ (it only contributes an overall phase). 
\end{enumerate}

Observation~(1) and~(2) follow directly from Table~\ref{tab2}, while Observation~(3) follows from Eq.~(4).
We now leverage these observations to derive a deeper insight.

\paragraph{Color-code structure and $Z$-constraints}
Observation~(1) implies that all non-Clifford gates act exclusively on the $19$ data qubits of a $[[19,1,5]]$ color code block.
Since this is a CSS stabilizer code, its stabilizer generators can be chosen as $r_X$ $X$-type checks and $r_Z$ $Z$-type checks, with
\begin{equation}
r_X + r_Z = n - k,
\end{equation}
where $(n,k)=(19,1)$. For the triangular distance-$5$ color code used in MSC, one may take $r_X=r_Z=9$, hence the data block carries $r_Z=9$ independent $Z$-type parity constraints.

\paragraph{A coset bound for diagonal $T$-layers}
Let $Q$ denote the set of qubits on which $T/T^\dagger$ may act, and let
$G_Q := \langle Z_q : q\in Q\rangle \cong (\mathbb{Z}_2)^{|Q|}$
be the group of $Z$-Pauli strings supported on $Q$.
For a given stabilizer component $\ket{\Psi_S}$ appearing in the generalized tableau, define its $Z$-type stabilizer subgroup on $Q$ by
\begin{equation}
H_Q(\Psi_S) := \mathrm{Stab}(\Psi_S)\cap G_Q,
\quad r_Q := \mathrm{rank}\bigl(H_Q(\Psi_S)\bigr).
\end{equation}
where $\mathrm{Stab}(\Psi_S)$ denotes the stabilizer group of the stabilizer component $\ket{\Psi_S}$.

\begin{proposition}[Coset bound for a $T$-layer]\label{prop:coset_bound}
Consider a layer of diagonal $\pi/4$ gates supported on $Q$,
$U := \bigotimes_{q\in Q} T_q^{\pm 1}$.
Then $U\ket{\Psi_S}$ admits a decomposition into at most
\begin{equation}
|G_Q/H_Q(\Psi_S)| = 2^{|Q|-r_Q}
\end{equation}
distinct stabilizer components (up to global phases). In particular, the generalized stabilizer tableau can be arranged so that the number of nonzero coefficients satisfies $|\mathbf{v}|\le 2^{|Q|-r_Q}$ immediately after applying $U$.
\end{proposition}

\begin{proof}[Proof sketch]
Using Eq.~(4), write $T_q^{\pm 1}=\alpha I+\beta Z_q$ with $\alpha,\beta\neq 0$.
Expanding the product yields
\[
U = \sum_{g\in G_Q} c_g\, g,
\qquad\text{hence}\qquad
U\ket{\Psi_S}=\sum_{g\in G_Q} c_g\, g\ket{\Psi_S}.
\]
If $g_1^{-1}g_2\in H_Q(\Psi_S)$ then $g_1\ket{\Psi_S}$ and $g_2\ket{\Psi_S}$ differ only by a phase because $H_Q(\Psi_S)\subseteq \mathrm{Stab}(\Psi_S)$.
Therefore distinct stabilizer components are indexed by cosets in $G_Q/H_Q(\Psi_S)$, whose cardinality is $2^{|Q|-r_Q}$.
\end{proof}

Proposition~\ref{prop:coset_bound} formalizes Observation~(3): if $Z_q\in \mathrm{Stab}(\Psi_S)$ (equivalently $Z_q\ket{\Psi_S}=\pm\ket{\Psi_S}$), then applying $T_q$ or $T_q^\dagger$ does not branch the decomposition and hence does not increase $|\mathbf{v}|$.

\paragraph{Why interleaving CNOT and $Z$-measurements does not violate the bound}
In the MSC construction, the measurement-and-reset subroutines between consecutive $T/T^\dagger$ layers are designed to repeatedly re-establish the code constraints on the data block (while extracting syndrome information and flag checks).
Consequently, before each $T/T^\dagger$ layer, every stabilizer component $\ket{\Psi_S}$ in the generalized tableau satisfies the $r_Z=9$ (for $d=5$ case) independent $Z$-type checks of the $[[19,1,5]]$ block (with signs determined by the measurement record), implying $r_Q \ge 9$ for $Q$ being the $19$ data qubits. For $d=3$ case, $|Q|=7$ and $r_Z = 3$.

The operations between $T/T^\dagger$ layers are Clifford gates (notably CNOT) together with $Z$-basis measurements and resets on ancillas.
Clifford gates map stabilizer states to stabilizer states bijectively and therefore permute the stabilizer components without increasing their count; in terms of the coefficient vector, they apply a relabeling (and possibly phase updates) but do not introduce new cosets beyond $G_Q/H_Q$.
$Z$-basis measurements act as projections: they can only remove components inconsistent with the observed outcome or merge components that become equivalent after enforcing the measurement constraint, and thus can only decrease (never increase) $|\mathbf{v}|$.

With these analysis, we obtain the following theorem:
\begin{theorem}[Upper bound of $|\mathbf{v}|$]\label{thm:coff_upper_bound}
    Consider simulating a magic-state cultivation (MSC) circuit using the generalized stabilizer tableau formalism for pure states. At every time step of the MSC simulation, the number of nonzero coefficients satisfies:
\begin{enumerate}
    \item {For code distance $d=3$,} 
    \[
    |\mathbf{v}| \;\le\; 2^{|Q|-r_Z} \;=\; 2^{7-3} \;=\; 16.
    \]
    \item {For code distance $d=5$,}
    \[
    |\mathbf{v}| \;\le\; 2^{|Q|-r_Z} \;=\; 2^{19-9} \;=\; 1024.
    \]
\end{enumerate}
\end{theorem}

Theorem~\ref{thm:coff_upper_bound} agrees perfectly with our experimental data.
In the distance-$5$ simulations, the memory required to store the generalized stabilizer tableau for a single shot is at most $200~\mathrm{KB}$.
Given that the H800 GPU provides $80~\mathrm{GB}$ of device memory, in theory no fewer than $400{,}000$ shots can reside concurrently on a single GPU, which enables a high degree of shot-level parallelism in the simulation.

\subsection{Parallel Implementation on GPU}
\label{subsec:gpu_impl}

To efficiently simulate magic state cultivation circuits, we implement the above generalized stabilizer tableau algorithm on modern GPUs.
Our CUDA backend follows a shot-parallel design: for a fixed circuit, we execute many independent shots concurrently on the GPU, and each shot maintains its own generalized stabilizer state $(\mathbf{v},B(\mathcal{S},\mathcal{D}))$ with no inter-shot communication.
All per-shot states remain resident in device memory throughout the simulation, and the host only reads back small aggregated statistics or error flags when necessary.

\paragraph{Memory layout}
For each shot, device memory is partitioned into four regions.
(1) \texttt{table} stores the binary symplectic stabilizer--destabilizer tableau for $(\mathcal{S},\mathcal{D})$.
(2) \texttt{entries} stores the sparse coefficient vector $\mathbf{v}$ as a list of nonzero pairs $(\alpha, v_\alpha)$, where $\alpha\in\{0,1\}^n$ indexes the basis states $|b_\alpha\rangle=d_\alpha|\psi_{\mathcal{S}}\rangle$.
(3) \texttt{decomp} provides temporary workspace for the vector updates (Pauli measurements and $T/T^\dagger$), e.g., intermediate indices, phases, and compaction buffers.
(4) \texttt{memory} stores classical per-shot data such as measurement outcomes, sampled branch probabilities, random number generator~(RNG) state, and error codes.
All buffers are allocated \emph{statically} before execution to avoid dynamic allocation overhead on the GPU; \texttt{entries} and \texttt{memory} may grow and shrink logically, but are implemented with fixed capacities sized to accommodate the entire circuit.
If a shot exceeds its allocated capacity (e.g., due to an unexpectedly large $|\mathbf{v}|$), we mark it as an overflow and optionally rerun it under a larger configuration.

\paragraph{Parallelization strategy}
We exploit fine-grained parallelism by mapping different operations to different kernel granularities.
For Clifford gates, tableau updates touch all $2n$ generators, so we parallelize over the Cartesian product of shots and tableau rows.
For Pauli measurements and $T/T^\dagger$ gates, the dominant work scales with $|\mathbf{v}|$, and we therefore implement these updates as short kernel pipelines: some stages launch one thread per nonzero entry, while others launch one thread per shot for normalization and bookkeeping.
All randomness is generated independently per shot using per-shot RNG instances and seeds.

\paragraph{Random seed generation}
To ensure that random seeds do not collide across the massive number of shots, we generate per-shot seeds using a SHA-1 based hashing scheme. 
This construction provides a large effective seed space and makes repeated seeds extremely unlikely, which is important for statistically sound Monte Carlo estimation over billions of shots.

\section{Results and Analysis}
\label{sec:results}
We have simulated more than 200 billion shots in total for different noise cases. The simulator runs on a small GPU cluster equipped with 16 NVIDIA H800 GPUs. All experiments were completed within 20 days of wall-clock time.

Our simulator is now open source, and readers can find the project on GitHub~\cite{SOFT2025} to reproduce all results reported in this paper. To minimize the impact of floating-point error, all experiments in this section use double-precision floating-point arithmetic.

\subsection{Simulation Results of Magic State Cultivation}
\label{subsec:msc_effectiveness}

For the $d=3$ case, the entire circuit can be simulated using a straightforward state-vector simulator~\cite{gidney2024magic}, we reproduce the $d=3$ results without further discussion here. We focus on the more challenging $d=5$ case and report results under several different physical noise levels, which cannot be achieved by existing open-source simulators.

Table~\ref{tab3} summarizes the performance of magic state cultivation for $d=5$ under different physical noise strengths $p$.
For each $p$, we report the discard rate (fraction of shots rejected by the detectors), the logical error rate after postselection and the total number of shots. Figure \ref{fig2} displays the logical error rates along with their statistical uncertainties.

\begin{table*}[t]
    \centering
    \caption{Performance of magic state cultivation for distance $d=5$ under different physical noise levels $p$ (1 B = 1 billion).}
    \label{tab3}
    \begin{tabular}{c!{\vrule width 1.2pt}c|c|c|c|c}
        \hline
        Noise Strength  & Total Shots & Preserved Shots & Discard Rate & Detected Logical Errors  & Logical Error Rate \\
        \Xhline{1.5pt}
        0.002  & 28.9B  & 0.60B  & 97.92\% & 22 & $3.41\times 10^{-8}$ \\
        \hline
        0.001  & 74.0B  & 10.6B & 85.60\% & 49 & $4.59\times 10^{-9}$  \\
        \hline
        0.0005 & 134.4B & 50.9B & 62.10\% & 8 & $1.57\times 10^{-10}$  \\
        \hline
    \end{tabular}
\end{table*}

We now highlight several results and observations:

\begin{enumerate}
    \item Since the structure of detection regions within the protocol is independent of whether the $|T\rangle$ or $|Y\rangle$ state is cultivated, the discard rates should remain invariant. Our numerical simulations corroborate this behavior. For $p \ge 0.003$, the discard rate exceeds $99\%$. In this regime, we consider magic state cultivation to be impractical. Consequently, we did not conduct additional simulations to estimate the logical error rate.
    \item At $p=1 \times 10^{-3}$, the logical error rate after postselection is approximately $4.59 \times 10^{-9}$, which is roughly $7.7$ times the conjectured value reported in~\cite{gidney2024magic}. For $p=5 \times 10^{-4}$, the logical error rate is $1.5 \times 10^{-10}$, or approximately $7.9$ times the conjectured value. These results demonstrate that the performance gap introduced by the Clifford proxy method scales with circuit size (from $d=3$ to $d=5$). This variation highlights the necessity of simulation tools capable of handling large circuits to rigorously validate the protocol's effectiveness.
    \item End-to-end simulations using a Clifford proxy for $d=5$ magic state cultivation yielded estimated logical error rates of $2 \times 10^{-9}$ and $4 \times 10^{-11}$ for physical error rates $p=1 \times 10^{-3}$ and $p=5 \times 10^{-4}$, respectively. Our results demonstrate that the first two stages alone exhibit a higher logical error rate than these proxy-based estimates. Consequently, the actual end-to-end logical error rate will inevitably exceed the previously reported values.
\end{enumerate}

\begin{table*}[t]
    \centering
    \caption{The discard rates at different noise strength.}
    \label{tab4}
    \begin{tabular}{c|c|c|c|c|c|c}
        \hline
        Noise Strength & 0.0005 & 0.001 & 0.002 & 0.003 & 0.004 & 0.005 \\
        \hline
        discard rate & 62.10\%  & 85.60\% & 97.92\%  &  99.70\%  & 99.96\%  &  99.99\% \\
        \hline
    \end{tabular}
\end{table*}

\begin{figure}
    \centering
    \includegraphics[width=0.9\linewidth]{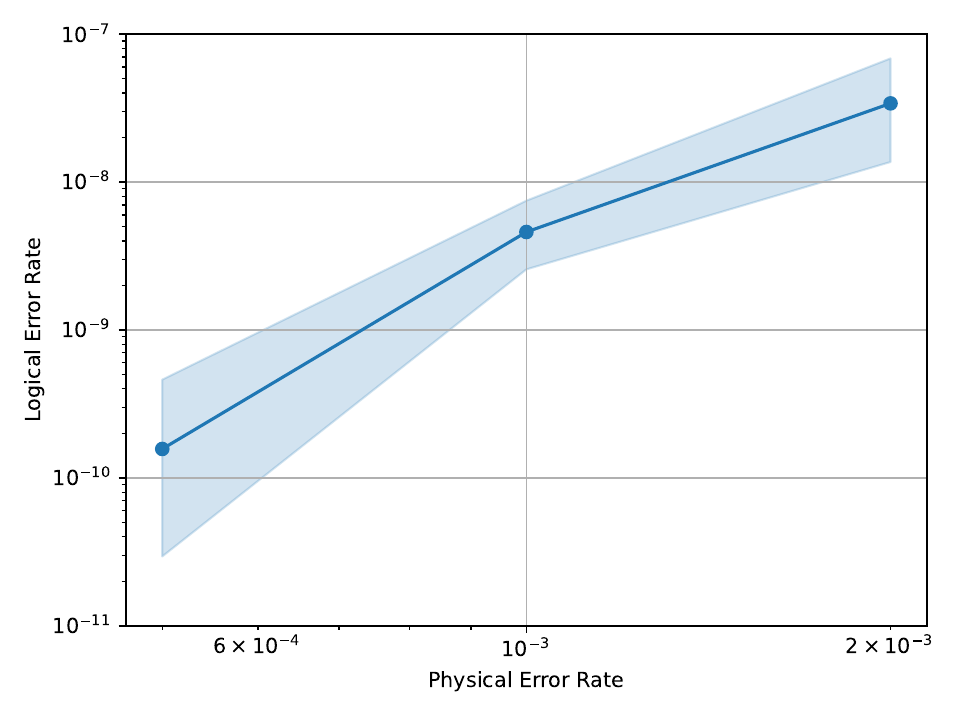}
    \caption{Logical error rate of $d=5$ magic state injection + cultivation at several physical error rates. Color highlights represent statistical uncertainty arising from Monte Carlo sampling. Each highlight spans hypothetical logical error rates that have a Bayes factor of no more than 1000 relative to the maximum likelihood hypothesis, under the assumption of a binomial distribution}
    \label{fig2}
\end{figure}

These results provide reliable numerical evidence of a gap between the conjectured and actual performance of the $d=5$ MSC protocol. Nevertheless, the reported logical error rates remain sufficiently promising to suggest the continued effectiveness of the protocol.

\subsection{Validation of Our Simulator}
\label{subsec:validation}
For $d=3$ benchmarks with $p=0.001$, our results and the data on Zenodo~\cite{GidneyZenodo} coincide, with a discard rate of $31.3\%$ and logical error rate of around $9\times10^{-7}$, not $35\%$ and $6\times10^{-7}$ given in the MSC paper~\cite{gidney2024magic}.

To validate the correctness of our GPU-based simulator, we compared its outputs against two independent implementations:

\begin{itemize}
    \item A C++ implementation of the generalized stabilizer formalism, which simulates Clifford+$T$ circuits using the same representation as the CUDA version.
    \item The general-purpose state-vector simulator and the extended stabilizer simulator provided by Qiskit.
\end{itemize}

We constructed a variety of \textit{small} Clifford+$T$ circuits and injected random Pauli noise into them.
For all circuit instances tested, the logical outcome distributions produced by our GPU simulator exactly matched those from the C++ implementation and from Qiskit's statevector simulator.
Moreover, during these cross-checks, we verified that at every simulation step the stabilizer tableau $B(\mathcal{S}, \mathcal{D})$ and the coefficient vector $\mathbf{v}$ in the GPU implementation were identical to those in the reference C++ implementation.

The source code of the C++ reference simulator is also included in our GitHub repository~\cite{SOFT2025}, so that independent users can reproduce these consistency checks.

To further increase confidence in our results for the $d=5$ MSC circuits, we provide over 10 explicit examples of rare logical FAILUREs on the ``ungrown'' $d=5$ circuit, that is, cases where the detector pattern passes but the final magic state is incorrect.
Such examples are intrinsically hard to find: if too many noise events occur, the shot will be quickly discarded by the detectors; if noise is very rare or absent, the error-correcting mechanism produces the correct magic state.
In our $d=5$ simulations, under $p=0.001$ and $p=0.002$, we typically had to simulate more than $10^9$ shots to encounter a single non-discarded logical error. These samples can also be found in \cite{SOFT2025}.

\subsection{Performance and Scaling Behavior}
\label{subsec:performance}
In this subsection, we study the performance and scaling behavior of our simulator.

{To show the performance advantage of SOFT, we compare it on $d=3$ and $d=5$ MSC circuits against two simulators: the extended-stabilizer backend in Qiskit (ES)~\cite{Qiskit_extend_stb, bravyi2019simulation} and stabilizer tensor network (STN) simulator in \cite{masot2024stn}. These two simulators are the open-source candidates we could find that satisfy:
\begin{itemize}
    \item it supports Clifford+T gates, noise, mid-circuit measurement and feedback gate (real-time control);
    \item it is based on the technique of stabilizer decompositions; 
\end{itemize}
Both ES and STN are run on a 2.6 GHz Intel Core i7 CPU. SOFT is run on an NVIDIA H800 GPU. The speedup ratio is shown in Table \ref{tab4}. Although comparing a GPU-based simulator with CPU-based simulators is not entirely fair, such a high speedup ($>10^{4}$ for all cases) indeed demonstrates the efficiency of our parallel implementation. Moreover, the speedup increases as the circuit complexity grows. We did not compare some state-vector simulators since they are good at circuits with less than $15$ qubits but they cannot handle $d=5$ MSC circuit with 42 qubits on a single CPU or GPU.}

\begin{table*}
    \centering
    \caption{The speedup of SOFT over comparison simulators. }
    \label{tab4}
    \begin{tabular}{c|c|c|c|c}
        \hline
        Circuits & SOFT's Average Time Per Shot & Comparison Simulator & Time Per Shot & Speedup ($T/T_{SOFT}$)
        \\
        \hline
        \multirow{2}{*}{d=3 cultivation} & \multirow{2}{*}{$6.68\times10^{-6}$ s} & Qiskit ES & 1.07s  &  $\mathbf{1.60\times10^{5}}$ \\
        \cline{3-5}
        & & STN & 0.21s & $\mathbf{3.14\times10^{4}}$ \\
        \hline
        \multirow{2}{*}{d=5 cultivation} & \multirow{2}{*}{$9.37\times10^{-5}$ s} & Qiskit ES & timeout ($\geq 1$ hours)  &  \textbf{N/A} \\
        \cline{3-5}
         & & STN & 59.8s & $\mathbf{6.38\times10^5}$ \\
        \hline
    \end{tabular}
\end{table*}

Second, we evaluate the scaling behavior of simulation throughput over \texttt{batch\_size} on a single GPU. 
Here, batch\_size denotes the number of independent shots that are launched \emph{concurrently} on one GPU; each shot is handled by one GPU kernel, and the GPU runtime schedules these kernels across streaming multiprocessors (SMs). As shown in Figure~\ref{fig3}, once batch\_size becomes comparable to the GPU's maximum number of resident threads, the throughput in shots per second approaches saturation. 
Further increasing batch\_size provides little additional benefit. The largest batch\_size (around 520,000) is chosen such that all concurrently active shots fit within the available H800 device memory, ensuring that the per-shot state and auxiliary buffers do not exceed the GPU memory capacity.

Finally, we show how the simulation throughput on a single GPU depends on noise strength $p$, depicted in Figure \ref{fig:thrpt_vs_noise}. From the figure we found that as $p$ increases the simulation throughput also increases. This is because the discard rate increases as $p$ increases. One shot will be immediately discarded and stopped if it failed some detector parity check, instead of simulating the whole circuit. On the other hand, as the number of noise events increases, a shot is more likely to be discarded earlier.

\begin{figure}
    \centering
    \includegraphics[width=0.80\linewidth]{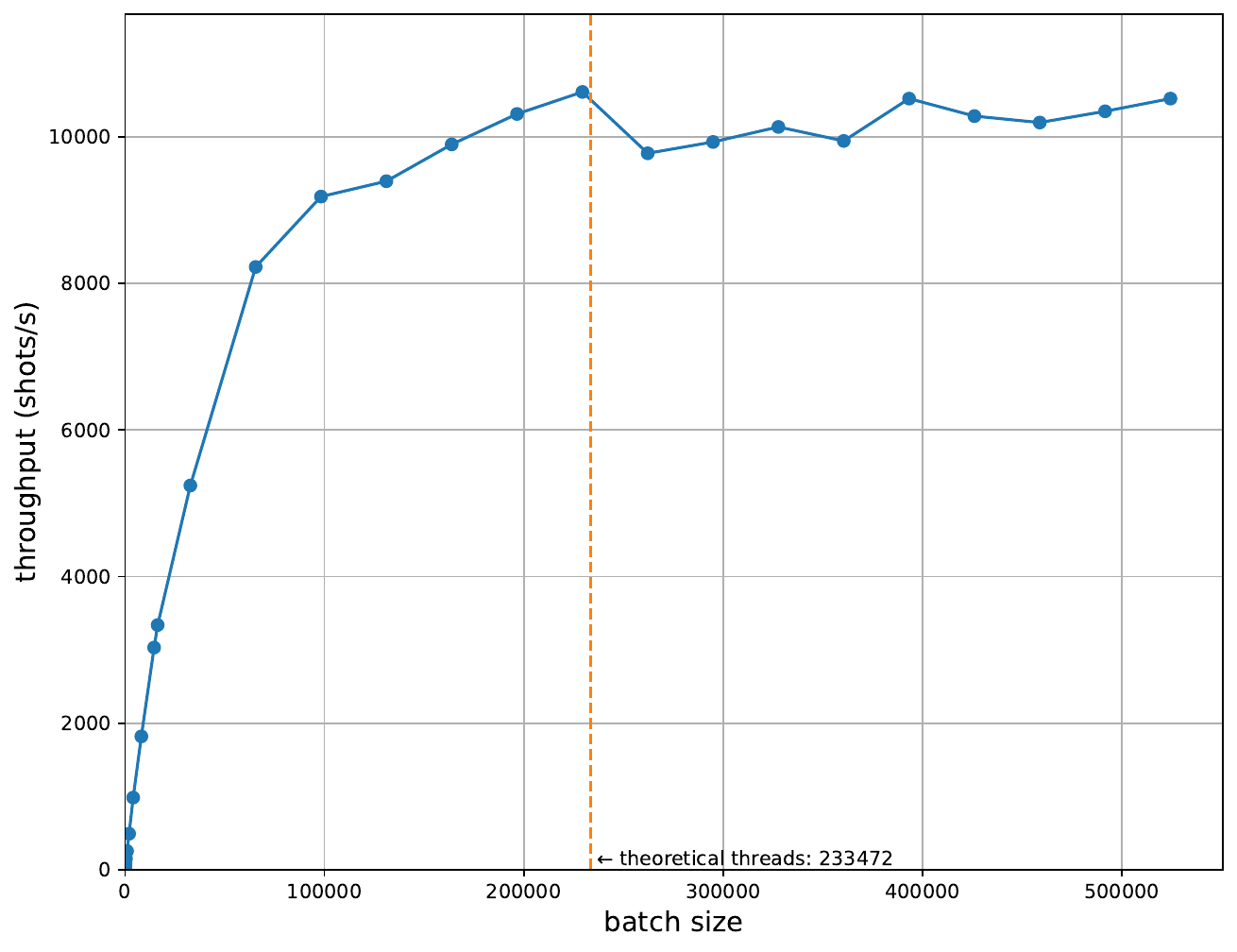}
    \caption{The sampling throughput versus \texttt{batch\_size} on a single H800 GPU, with $d$=5 MSC circuit and $p=0.001$. Here, orange dashed line means the theoretical maximum number of resident threads of H800 PCIe: 114 $\mathrm{SMs}\ \times$ 2048 threads/\rm{SM} = 233,472 threads.}
    \label{fig3}
\end{figure}

\begin{figure}
    \centering
    \includegraphics[width=0.80\linewidth]{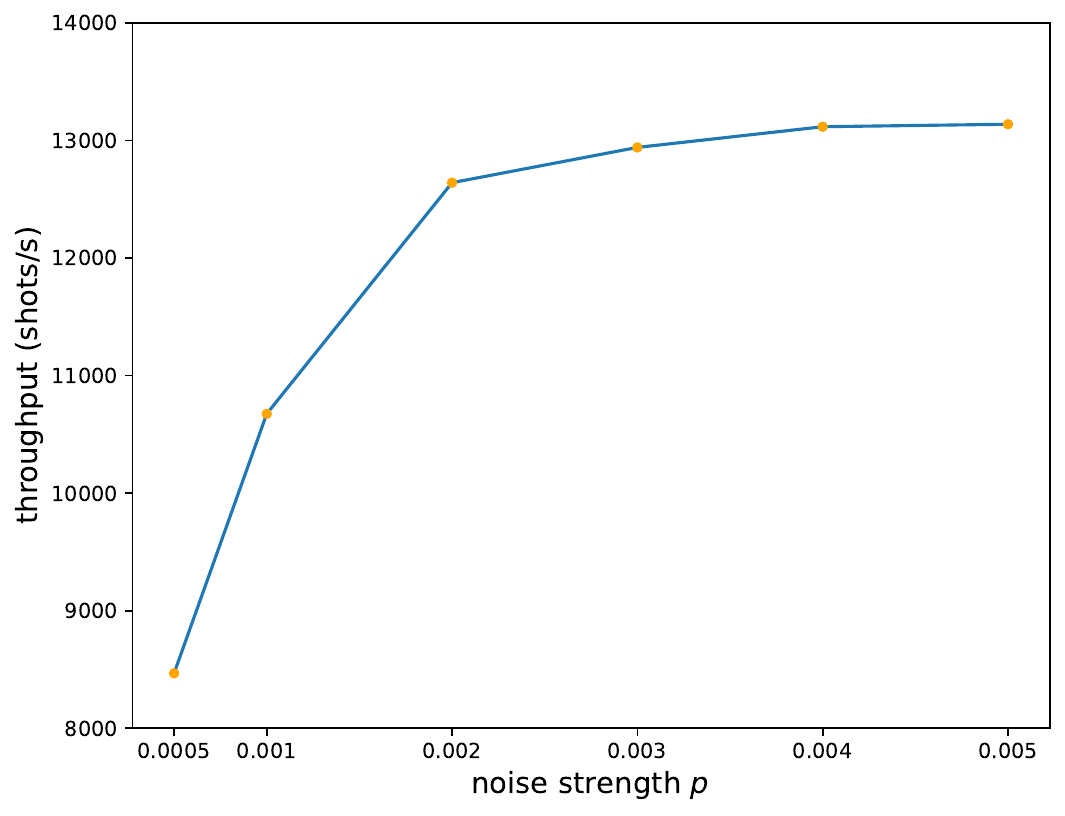}
    \caption{The sampling throughput versus noise strength $p$ on a single H800 GPU, with code distance $d$=5. During these experiments, we always set the \texttt{batch\_size} to be $2^{18} = 262,144$.}
    \label{fig:thrpt_vs_noise}
\end{figure}

\section{Conclusion}
\label{sec:conclusion}

We have presented SOFT, a high-performance simulator for universal fault-tolerant quantum circuits operating in the presence of circuit-level noise and mid-circuit measurements. 
It achieves a speedup of several orders of magnitude over existing simulators,
significantly extending the computational reach for the simulation of noisy Clifford+T circuits at scales previously considered intractable. 
This leap in simulation capability marks a pivotal shift from simulating quantum memory to simulating true quantum computation, making SOFT an valuable tool for designing resource-efficient fault-tolerant architectures.

As a concrete application, we deployed SOFT to simulate the distance-5 magic state cultivation circuit. 
Our simulation confirms the efficacy of the protocol in generating high-fidelity logical $T$ states but uncovers a significant discrepancy in the logical error rate compared to predictions derived from Clifford proxies. 
We thus resolve a critical uncertainty in the protocol's design, underscoring the necessity of reliable non-Clifford verification and prediction.

Although we have focused on Clifford+$T$ circuits in this work, our implementation strategy can be straightforwardly applied to other non-Clifford gates, such as the Toffoli gate. 
Moreover, beyond the state preparation, our simulator can also be extended to incorporate gate teleportation. We expect the simulation complexity to remain comparable, as the non-Clifford gates are concentrated within the state generation phase, while the teleportation protocol itself does not increase the non-Clifford gate count. 
Furthermore, our systems-level GPU implementation offers a highly extensible framework; it provides significant headroom for low-level performance tuning while remaining architecturally amenable to integrating higher-level algorithmic optimizations.

\section*{Acknowledgement}
We thank Wang Fang, Minpu Qin, and Pan Zhang for their helpful discussions, and acknowledge the useful feedback from StackExchange~\cite{stackEx}.

\bibliographystyle{IEEEtran}
\bibliography{references}

\end{document}